\documentclass{article} 
\usepackage{authblk}
\usepackage[table]{xcolor}
\usepackage{times}
\usepackage{eso-pic} 
\usepackage[margin=1.1in]{geometry}
\usepackage{authblk}

\usepackage[square,numbers]{natbib}
\usepackage[utf8]{inputenc} 
\usepackage[T1]{fontenc}    
\usepackage{hyperref}       
\usepackage{url}            
\usepackage{booktabs}       
\usepackage{amsfonts}       
\usepackage{nicefrac}       
\usepackage{microtype}      
\usepackage{xcolor}         
\usepackage{url}            
\usepackage{booktabs}       
\usepackage{amsfonts}       
\usepackage{nicefrac}       
\usepackage{microtype}      
\usepackage{xcolor}         
\usepackage{placeins}       
\usepackage[normalem]{ulem}

\usepackage{amsmath,amssymb,amsthm,bbm,color,nccmath}
\usepackage{graphicx}
\usepackage[capitalize,noabbrev]{cleveref}
\usepackage{algorithm}
\usepackage{algpseudocode}
\usepackage{caption}
\usepackage{subcaption}
\usepackage[export]{adjustbox}
\usepackage{afterpage}
\usepackage[hang,flushmargin]{footmisc}
\usepackage{float}
\usepackage{enumitem,kantlipsum}
\usepackage{multirow}
\usepackage{makecell}
\usepackage{array}
\usepackage{amsthm}
\usepackage{wrapfig}

\newtheorem{prop}{Proposition}

\theoremstyle{plain}
\newtheorem{theorem}{Theorem}[section]

\theoremstyle{definition}
\newtheorem{definition}[theorem]{Definition}

\theoremstyle{remark}

\usepackage[export]{adjustbox}

\DeclareMathOperator{\rgt}{\textbf{rgt}}
\DeclareMathOperator{\maxrgtlv}{\ell_{\rgt}}

\newcommand{\ourmethod}{\texttt{SPACE\,}}

\title{Statistically Truthful Auctions via Acceptance Rule}
\author[1]{Roy Maor Lotan}
\author[2]{Inbal~Talgam-Cohen}
\author[1,2]{Yaniv Romano}
\affil[1]{Department of Electrical and Computer Engineering, Technion IIT, Israel}
\affil[2]{Department of Computer Science, Technion IIT, Israel}
\date{}

\begin{document}

\maketitle

\begin{abstract}
Auctions are key for maximizing sellers' revenue and ensuring truthful bidding among buyers. Recently, an approach known as differentiable economics based on machine learning (ML) has shown promise in learning powerful auction mechanisms for multiple items and participants. However, this approach has no guarantee of strategy-proofness at test time. Strategy-proofness is crucial as it ensures that buyers are incentivized to bid their true valuations, leading to optimal and fair auction outcomes without the risk of manipulation. In this work, we propose a formulation of statistical strategy-proofness for auction mechanisms. Specifically, we offer a method that bounds the regret---quantifying deviation from truthful bidding---below a pre-specified level with high probability. Building upon conformal prediction techniques, we develop an auction acceptance rule that leverages regret predictions to guarantee that the data-driven auction mechanism meets the statistical strategy-proofness requirement with high probability. Our method---\emph{\textbf{S}tatistically \textbf{T}ruthful Auctions via \textbf{A}cceptance \textbf{R}ule (\ourmethod)}---represents a practical middle-ground between two extremes: enforcing truthfulness---zero-regret---at the cost of significant revenue loss, and naively using ML to construct auctions with the hope of attaining low regret, with no test-time guarantees.
\end{abstract}

\begin{figure*}[ht]
    \centering
    \includegraphics[width=0.95\linewidth]{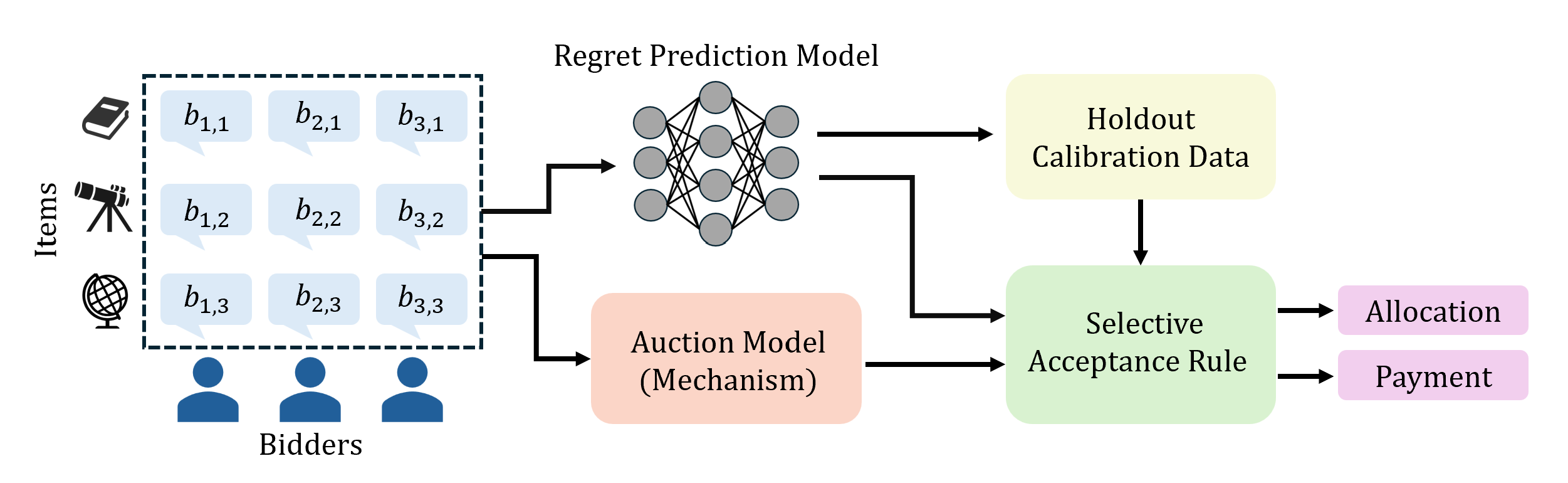}
    \caption{
    Overview of \emph{Statistically Proofed Auctions via Conformal Estimation (\ourmethod)}. 
    Bidders submit bids for multiple items, forming a bid matrix. An auction model (mechanism) processes these bids to determine allocation and payment. A regret model estimates the incentive to deviate from truthful bidding. Using regret predictions and holdout calibration data, a selective acceptance rule ensures statistical strategy-proofness, accepting only outcomes with low estimated regret. This balances revenue and incentive compatibility.
    }
    \label{fig:method_illustration}
\end{figure*}

\section{Introduction}\label{sec:intro}
Auctions are a core tool of economic theory and practice, allocating valuable 
goods and generating billions across markets such as online advertising and procurement; 
notable implementations include sponsored search (e.g., Google) and eBay auctions \citep{milgrom2004putting, krishna2009auction, cramton1997fcc, sandholm2007expressive, leyton2017economics, milgrom2017discovering, roth2018marketplaces, edelman2007internet}. In the standard \emph{independent private valuations} model, each bidder has a \emph{valuation function} over subsets of items drawn independently, possibly from distinct distributions, the auctioneer knows these distributions, yet valuations remain private—creating incentives to misreport and complicating design.

Optimal mechanisms are known only in several special cases such as Myerson’s setting---single-item multi-bidder auctions~\citep{Myerson81}. For multiple items with a single bidder, much of the landscape is characterized \citep{Manelli2006Bundling,Pavlov2011Optimal,Daskalakis2017Strong,Kash2016Optimal}. For multiple items and multiple bidders, optimal mechanisms are known only in simplified settings \citep{Yao2017Dominant}, and their characterization
remains open even with a few bidders and several items. Leveraging the predictive power of deep neural networks (DNN), \citet{dutting2019optimal} and follow-ups \citep{Dutting0NPR24, duan2024scalable, rahme2021permutation, rahme2021auction, ivanov2022optimal, peri2021preferencenet,curry2022differentiable, wang2024gemnetmenubasedstrategyproofmultibidder, shen2019automated} propose data-driven mechanisms that, while lacking guarantees of revenue optimality or full strategy-proofness, empirically approximate known optimal designs when these are known. 

\subsection{The Challenge of Strategy-Proofness in Data-Driven Auctions}

Strategy-proofness is tightly linked to regret in auctions: regret measures the incentive of bidders to deviate from the truth, so a mechanism with low regret leaves bidders little reason to misreport their true valuations. Put differently, low regret promotes truthful bidding, yielding at the limit a strategy-proof auction mechanism
\citep{DuettingFJLLP12,milgrom2004putting, Myerson81}. The practical stakes are clear: lower regret supports greater trust in the mechanism, since participants can be confident that bidding truthfully is roughly in their best interest, and thus enables the auctioneer to guarantee higher revenue and more efficient allocation of resources across bidders \citep{krishna2009auction}.

A critical deployment issue for many data-driven---DNN-based---mechanisms is the absence of statistical guarantees at test time for new auctions. Although these models are trained to minimize regret on average and often achieve low-regret outcomes, they can still output allocations and payments that severely violate the low-regret requirement on unseen instances. This limitation is illustrated in Figure~\ref{fig:rgt_density}: the well-known RegretNet mechanism \citep{dutting2019optimal} does not consistently attain low regret, with a visible tail of high-regret outcomes (the red zone). While average regret may be small across many scenarios, the tail behavior persists. Figure~\ref{fig:rgt_vs_num_auc} emphasizes this point by showing that the maximum test-time regret rises as more auctions are evaluated, suggesting diminished reliability as the model encounters a broader range of auction contexts.

\begin{figure}[ht]
    \centering
    \begin{subfigure}[b]{0.45\textwidth}
        \centering
        \includegraphics[height=4.55cm]{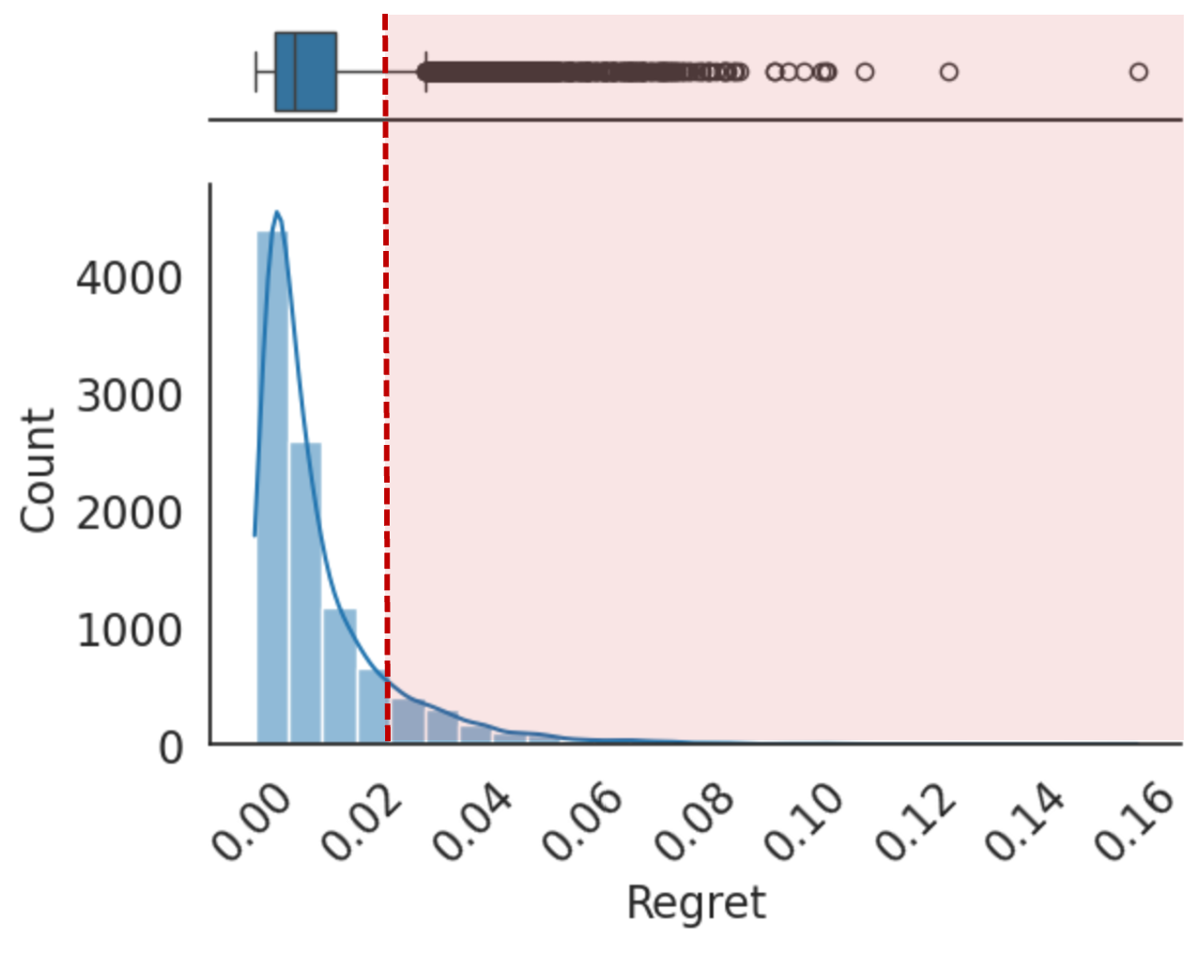}
        \caption{Empirical distribution of actual regret values of \texttt{RegretNet} model, evaluated on 10,000 test auctions.}
        \label{fig:rgt_density}
    \end{subfigure}
    \hfill
    \begin{subfigure}[b]{0.45\textwidth}
        \centering
        \includegraphics[height=4.52cm]{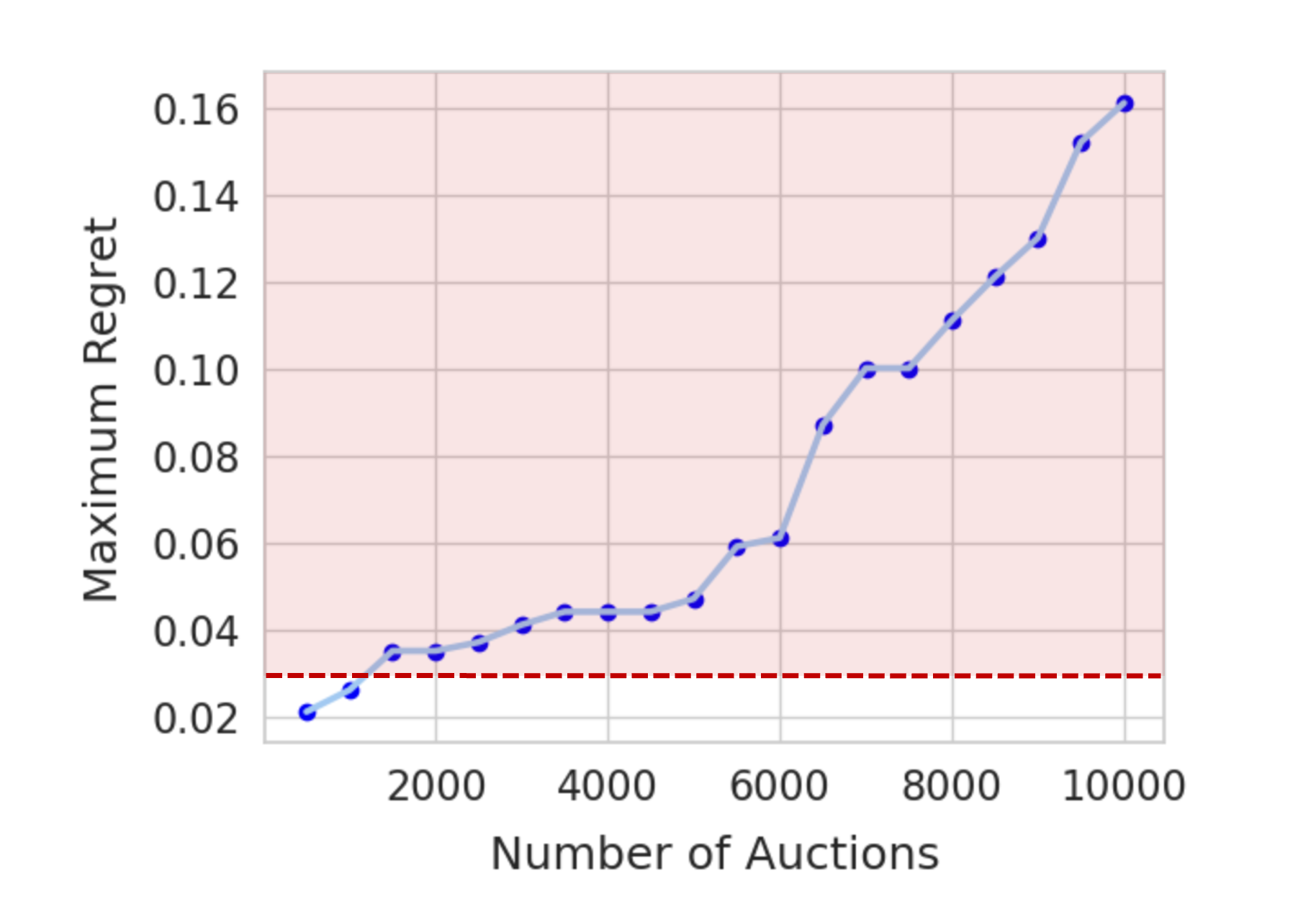}
        \caption{Maximum regret value obtained by \texttt{RegretNet} model as a function of the number of test auctions.}
        \label{fig:rgt_vs_num_auc}
    \end{subfigure}
    \caption{Regret behavior of \texttt{RegretNet} in the $2$-bidder, $3$-item auction setting.
    For additional details on the experimental configuration of this model, please refer to Section~\ref{sec:exp_setup}.}
    \label{fig:regret_combined}
\end{figure}

\vspace{-0.5em}

The discussion above reveals a significant limitation of cutting-edge data-driven auction design: there is no assurance that the regret will remain low with any level of confidence once the model is deployed in real-world scenarios. This undesired behavior undermines the trustworthiness of the auction system, and could lead to inefficiency, substantial revenue losses for the auctioneer, and unfair outcomes for the bidders. The primary goal of this paper is to tackle this challenge by selectively accepting auctions for which the unknown test-time regret is unlikely to exceed a desired maximum regret level.

\subsection{Related Work}
Given the theoretical and practical importance of strategy-proofness, there has been growing activity in enhancing the reliability of data-driven auctions. \citet{curry2020certifying} modify the model proposed by~\citep{dutting2019optimal} to enable certifiable strategy-proofness using mixed integer programming. However, the certification procedure imposes strict constraints both on the architecture design and the number of trainable parameters---as it involves solving a large mixed integer programming problem---limiting the model's performance. By contrast, as we will see, the framework we propose in this paper can be applied to virtually any auction model, regardless of its complexity and architectural design.

A different line of work by \citet{peri2021preferencenet,rahme2021auction,rahme2021permutation,ivanov2022optimal} offers novel learning approaches and architectures to formulate actions with low average regret. This stands in striking contrast with our method, which achieves a high probability bound on the worst-case, maximal regret obtained at test time.

Another line of work for forming reliable data-driven auctions builds on the parameterization of allocation menus that are guaranteed to be strategy-proof. This area was recently advanced by \citet{curry2022differentiable} and \citet{duan2024scalable}, who optimize the parameters of affine maximizer auctions (AMA)~\citep{AMA}. While AMA-based approaches are strategy-proof, they often yield lower revenue than standard DNN-based methods. Indeed, in Section \ref{sec:experiments}, we show that the approach by \citep{duan2024scalable} results in a revenue loss of about 45\% compared to our method. This emphasizes the hardness of attaining exact strategy-proof auctions while constructing data-driven auctions with high revenue.

\subsection{Overview of Contributions}
The challenge of forming guaranteed low-regret auctions with high revenue drives us to propose a new, statistical approach for constructing reliable data-driven auctions. Our key contribution is the introduction of an auction acceptance rule that selectively rejects predicted auctions whose regret exceeds a pre-defined threshold. In turn, the advantage of our approach is that it guarantees, with high probability (e.g., 99\%), that accepted auctions at test time do not exceed a pre-defined max regret value (e.g., of level 0.05).

Our approach consists of three key components. First, we present a novel, data-driven regret model to estimate the regret of an out-of-sample auction. We then build on conformal prediction to form a rigorous acceptance rule, ensuring that only auctions meeting the desired strategy-proofness level are accepted. Under the i.i.d. assumption, we prove that our approach controls the maximum regret level at test time, for any complex auction model, any unknown bidding distribution profiles, and regardless of the accuracy of the regret estimation model used to form the acceptance rule.

As such, our work introduces a new practical paradigm for balancing the trade-off between strategy-proofness and revenue discussed above. Our approach can be seen as a synthesis of two extremes: (i) enforcing zero regret at the expense of significant revenue loss, and (ii) naively using machine learning to design auctions with the hope of achieving low regret at test time. Indeed, our experiments show that the accepted auctions obtained by our approach are of low regret while having revenue nearly the same as the baseline auction model by \cite{dutting2019optimal}. As illustrated in Figure~\ref{fig:revenue_comparison}, our statistically strategy-proof mechanism achieves a favorable balance between revenue and strategy-proofness, outperforming traditional strategy-proof models in terms of revenue while maintaining controlled regret.

\begin{figure}[ht]
    \centering
    \includegraphics[height=4.52cm]{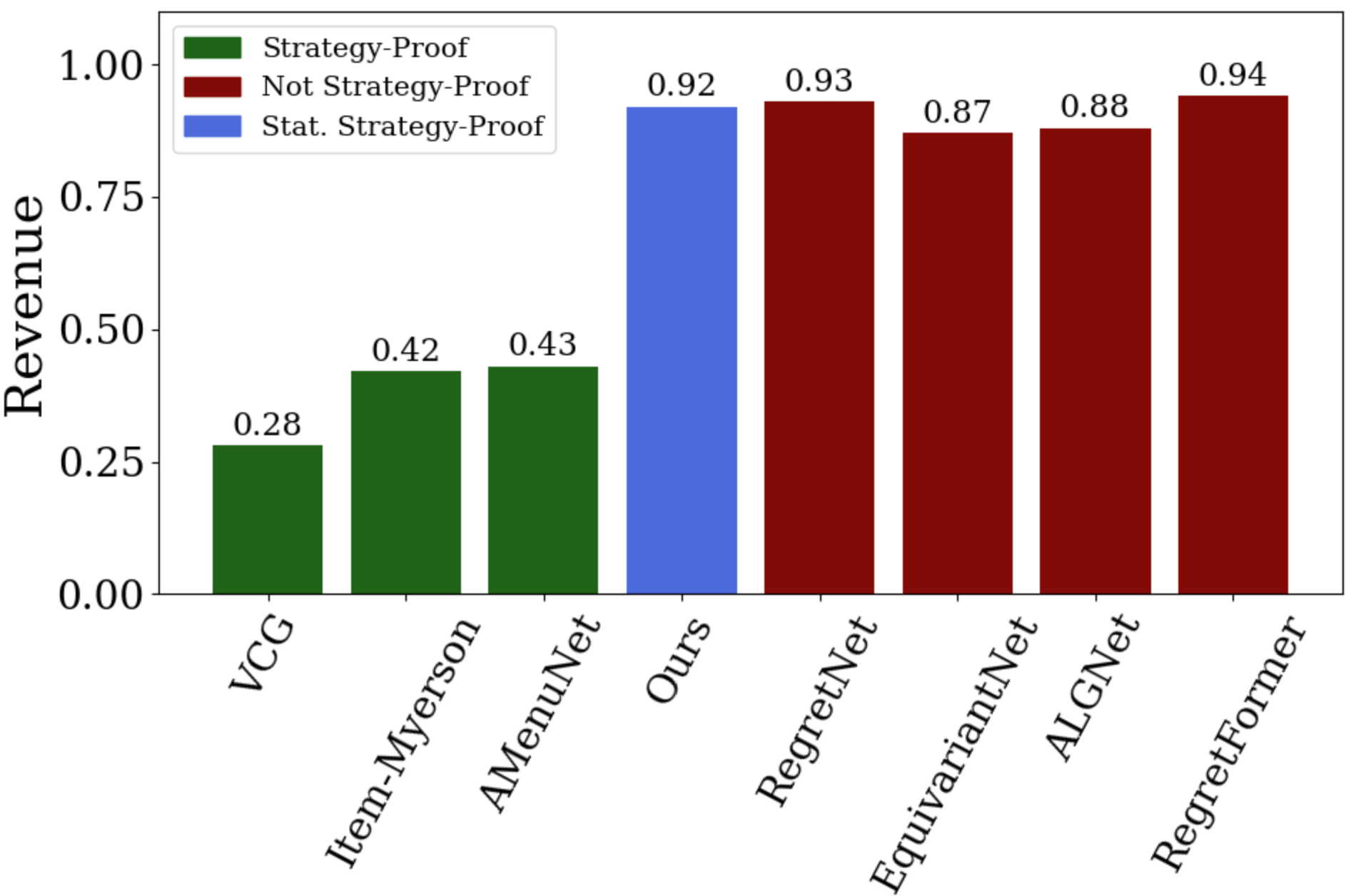}
    \caption{Comparison of revenue across different auction mechanisms. Strategy-proof methods (green) ensure truthful bidding, but yield lower revenue. Non-strategy-proof methods (red) achieve higher revenue at the cost of truthfulness. Our proposed statistically strategy-proof mechanism (blue) balances revenue and strategy-proofness. See Table \ref{tab:results_rev} for references.}
    \label{fig:revenue_comparison}
\end{figure}

\section{Background}
\label{sec:background}

\subsection{Problem Setup}
\label{sec:problem_setup}
We consider an auction environment, comprising of a set of $n$ bidders, denoted by $[n] = \{1, \ldots, n\}$, and a set of $m$ items, denoted by $[m] = \{1, \ldots, m\}$. Each bidder $i\in [n]$ possesses a private valuation $v^i \in \mathbb{R}^m$ which maps item sets $S\subseteq [m]$ to their value $v^i(S)\in \mathbb{R}_{\ge 0}$ for the bidder. The valuation matrix $v = [v^1, \ldots, v^n]$ is assumed to be drawn from a distribution $P$.

In an auction, participants submit bids $b^i\in \mathbb{R}^m_{\ge 0}$ that may not necessarily represent their true valuations. 
The bid profile of all bidders is represented by $b \in \mathbb{R}_{\geq 0}^{n \times m}$. 
The auction mechanism, denoted by $f(b) = [g(b), p(b)],$ comprises of an allocation rule $g(b)$ and a payment rule $p(b)$. Specifically, the allocation rule is defined as $g: \mathbb{R}_{\geq 0}^{n \times m} \to [0,1]^{n \times m},$
determining how the $m$ items are assigned to the $n$ bidders. The payment rule is given by $p: \mathbb{R}_{\geq 0}^{n \times m} \to \mathbb{R}_{\geq 0}^{n \times m},$ specifying the amount each bidder must pay for each item.
The primary goal of the auction mechanism is to allocate the items while setting payments, often with the objective of maximizing revenue.

To effectively frame strategy-proof auctions as a learning task, we shall design a DNN trained to minimize the regret while maximizing the bidders' utility. To this end, we denote the trained auction mechanism by  ${f}_{\hat{\theta}}(b) = [{g}_{\hat{\theta}}(b), {p}_{\hat{\theta}} (b)]$, where the vector $\theta$ represents trainable parameters. 
Here, the `hat' symbol over the parameters $\hat{\theta}$ indicates that these are fixed, i.e., the model is pre-trained. Using the notations from \citep{dutting2019optimal}, the utility of bidder $i$, denoted by $u^i_{\hat \theta}(v^i; b)$, is defined as the difference between their valuation of the allocated items and the payment made, i.e., 
\begin{equation}
u_{\hat{\theta}}^{i}(v^i;b) = v^i({g}^{i}_{\hat{\theta}}(b)) - {p}^{i}_{\hat{\theta}}(b).
\end{equation}
Above, $g^i_{\hat \theta}(b) \in \mathbb{R}$ and $p^i_{\hat \theta}(b)\in \mathbb{R}$ represent the allocation and payment for bidder $i$ given the bid profile~$b$. 

Using this utility function, the regret of bidder $i$ is defined as the difference in utility between bidding truthfully and bidding to maximize utility. Let $v^{-i}$ denote the valuation profile excluding bidder $i$. Then, the regret function of the $i$-th bidder is given by
\begin{equation}
    \label{def:regret}
    \rgt_{\hat{\theta}}^i(v;b) = 
    \max_{b^i \in b} u_{\hat{\theta}}^{i}(v^i;(b^i, v^{-i})) - u_{\hat{\theta}}^{i}(v^i;(v^i, v^{-i})).
\end{equation}

With valuations function at hand, choose mechanism parameters---train a DNN model---to maximize expected revenue, minimize the regret (strategy-proofness), subject to individual rationality, and feasibility. 
For further details, see \citet{dutting2019optimal}.

A notable feature of data-driven auctions is the prevalence of rejection, where the auction mechanism chooses not to allocate any items---formally denoted as allocation $A_{0}$ \cite{baisa2017auction, milgrom2013deferred, shui2023rejection}. Such no-allocation outcomes are common in practice, particularly in online platforms like ad auctions or marketplaces, where not every auction leads to a match between supply and demand. Rejection serves multiple purposes: it avoids inefficient allocations when the bids do not meet a reserve price or quality threshold, preserves fairness among participants, and can be used strategically to maintain long-term revenue or user experience. 

\subsection{Conformal Prediction}
\label{sec:conformal_prediction}
We now present the conformal prediction method, which we employ for controlling max regret levels within data-driven auction mechanisms. 

Conformal prediction~\citep{papadopoulos2002inductive,vovk2005algorithmic,lei2014distribution} offers a generic approach, applicable to any predictive model, for quantifying prediction uncertainty. In recent years, conformal prediction has increasingly gained traction across various domains, demonstrating its broad utility and effectiveness in practical settings \citep{angelopoulos2021gentle}. This method utilizes a subset of holdout \emph{calibration data} to quantify the prediction error of the model. These errors are then used to construct a range of plausible predictions---a \emph{prediction interval}---that covers the unknown test outcome with high probability. 

In what follows, we describe in more detail how to construct prediction intervals for regression problems, as this task is tightly connected to our goal of reliably predicting unknown regret values at test time. Let $x\in\mathcal{X}$ represent the input features to the model, $y\in\mathcal{Y}$ denote the true output, and ${f}_{\hat{\theta}}(x)$ be the predicted output generated by the model ${f}_{\hat{\theta}}$ parametrized by ${\hat{\theta}}$. The process of conformal prediction for general inputs and outputs can be outlined as follows. First, we apply ${f}_{\hat{\theta}}$ on the labeled calibration data $\mathcal{D}_{\textup{cal}} = \{(x_k, y_k)\}_{k=1}^K$ to construct a heuristic notion of prediction error for each calibration point $(x_k, y_k)$. We refer to this notion of error as a nonconformity score, denoted by $s(x,y)\in\mathbb{R}$. The next step involves calculating $\hat{q}_{\alpha}$, the $(1-\alpha)(1+1/K)$ empirical quantile of the calibration scores $s(x_k,y_k)$, $k=1,\dots,K$. This quantile is used to construct a prediction interval for a new test point $x_{\textup{test}}$, which includes all the output values $\tilde{y}\in\mathcal{Y}$ whose corresponding nonconformity scores fall below the quantile $\hat{q}_\alpha$, i.e. 
\begin{equation}
    \label{eq:conformal_set}
    C(x_{\textup{test}}) = \{\tilde{y}\in\mathcal{Y} : s(x_{\textup{test}},\tilde{y}) \leq \hat{q}_\alpha\}.
\end{equation}

The key property of $C(x_{\textup{test}})$ is that it is guaranteed to cover the unknown test label $y_\textup{test}$ at the desired level of coverage $1-\alpha$, as formally stated in Theorem~\ref{thm:conformal_calibration}. 
\begin{theorem}[Conformal Coverage Guarantee; Vovk, Gammerman, and Saunders~\citeyear{vovk1999machine}]
    \label{thm:conformal_calibration}
    Assume that the calibration set ${\mathcal{D}_\textup{cal} = \{x_i, y_i\}}_{k = 1}^{K}$ and the test point $(x_\textup{test}, y_\textup{test})$ are i.i.d. samples from $P_{x,y}$. Then, the following coverage guarantee holds for $\mathcal{C}(x_\textup{test})$ defined in \eqref{eq:conformal_set}:
    \begin{equation}
    \mathbb{P} \Big[ y_\textup{test} \in \mathcal{C}(x_{\rm test}) \Big] \ge 1 - \alpha.
    \end{equation}
\end{theorem}

Notably, the standard conformal prediction approach cannot be applied ``as is'' in auction settings. This is because the regret value, which we aim to control, depends on the user's bids $b$, valuations $v$, and the allocation and payment determined by the auction model ${f}_{\hat{\theta}}$. In other words, the ``label'' $y$ is now a complex function of both the model and the data, requiring a customized adaptation of conformal prediction to fit this auction context---this is the focus of our work.

\section{Proposed Method}
\label{sec:methods}
In this section, we present our approach to forming data-driven auctions with controlled maximal regret at test time. Given a test bidding matrix $b_{\textup{test}}$ and an unknown valuation matrix $v_{\textup{test}}$, we define the maximal regret as follows:
\begin{equation}
\rgt_{\hat{\theta}}^{\textup{max}}(v_\textup{test};b_\textup{test}) = \max_{i \in [n]}  \ \rgt^{i}_{\hat{\theta}}(v_\textup{test};b_\textup{test}), 
\end{equation}
where $\rgt^{i}_{\hat{\theta}}$ is defined in \eqref{def:regret}. 

We also define the calibration data set $\mathcal{D}_\textup{cal} = \{( b_k, v_k) \}_{k=1}^K$, which contains holdout pairs of bidding $b_k \in \mathbb{R}^{n \times m}$ and valuation $v_k \in \mathbb{R}^{n \times m}$ matrices that correspond to the $k$-th instance. Throughout this paper, we assume that the calibration samples $( b_k, v_k), k=1,\ldots,K$ and the test sample $(b_\textup{test}, v_\textup{test})$ are drawn i.i.d. from some unknown distribution $P_{b,v}$. 

Our approach consists of two key components. First, we develop a data-driven regret model, whose goal is to predict the unknown maximum regret of a new auction. Next, we show how to harness this regret estimator to construct an auction acceptance rule that accepts test actions that satisfy the following definition.
\begin{definition}[Statistically Strategy-Proof Auctions] \label{def:statistical_SP}
An auction mechanism ${f}_{\hat{\theta}}(b_\textup{test})=[{g}_{\hat{\theta}}(b_\textup{test}), {p}_{\hat{\theta}}(b_\textup{test})]$ is said to be {statistically strategy-proof} at a user-specified confidence level $(1-\alpha) \in [1,0)$ and maximal regret level $\maxrgtlv$ if 
\begin{equation}
\label{eq:statistically_strategy_proof_def}
\mathbb{P}\left[ \textup{\textbf{rgt}}^{\textup{max}}_{\hat{\theta}}(v_\textup{test};b_\textup{test}) > \maxrgtlv \right] \leq\alpha,
\end{equation}
where the probability is taken over calibration and test points.
\end{definition}

In words, an auction outcome that satisfies Definition \ref{def:statistical_SP} is guaranteed to have low maximum regret with high probability, where the maximal regret value $\maxrgtlv$ and the confidence level $1-\alpha$ are pre-defined by the auctioneer. Since we can control the maximal regret value with high probability, we refer to it as a \emph{statistically strategy-proof} auction mechanism. Notably, this statistical notion of strategy-proofness does not imply zero regret at test time.

\subsection{Regret Estimation}
\label{sec:regret_estimation}
Since the regret of a new auction is unknown, we should introduce a \emph{regret estimation model}. In its most general form, we formulate it as a neural network ${r}_{{\hat{\theta}}_r}(b)$, parameterized by ${\hat{\theta}}_r$; this model predicts $$\rgt_{\hat{\theta}}(v;b) = [\rgt^1_{\hat{\theta}}(v;b), \rgt^2_{\hat{\theta}}(v;b), \ldots, \rgt^n_{\hat{\theta}}(v;b)],$$
the outcome regret for a given auction. To fit the regret model on the training set $\mathcal{D}_\textup{train}$, 
we minimize the following objective function:
\begin{equation}
    \begin{aligned}
        \min_{\theta_r \in \mathbb{R}^d} \sum_{(b,v)\in \mathcal{D}_\textup{train}} \| r_{\theta_r}(b) - \rgt_{\hat{\theta}}(v;b) \|_1. 
    \end{aligned}
\end{equation}

We propose two different approaches to design the regret estimation model ${r}_{{\hat{\theta}}_r}(b)$. In the \emph{shared backbone setup}, the regret model is integrated into the existing allocation and payment backbone, providing a unified structure. The advantage of this approach is that it leverages the backbone features of the auction model to better predict the regret. However, this setup requires access to the auction model parameters and modifies the training procedure. In the second \emph{black-box} approach, the auction model is treated as a black-box, avoiding the need to access the auction model's internal parameters. Illustrations of both setups are provided in Appendix~\ref{asec:regret_model_illustration}.

\subsection{Model Calibration}
\label{sec:model_calibration}
In this section, we introduce a data-driven auction acceptance rule that accepts auctions with regret outcomes below $\maxrgtlv$---the maximum regret level specified by the auctioneer. A naive approach would be to accept a new auction ${f}_{\hat{\theta}}(b_{\text{test}})$ if it satisfies ${r}_{\hat{\theta}_r}^\textup{max}(b_{\text{test}}) \leq \maxrgtlv$, where  ${r}_{\hat{\theta}_r}^\textup{max}(b_{\textup{test}}) = \max_{i \in [n]} {r}_{\hat{\theta}_r}^\textup{i}(b_{\textup{test}}) $. In plain words, the auction is accepted if its estimated maximal predicted regret falls below the threshold $\maxrgtlv$. However, the key issue with this naive approach is that the regret estimator may be inaccurate, potentially leading to the acceptance of auctions whose actual regret $\rgt^{\textup{max}}_{\hat{\theta}}(v_{\text{test}};b_{\text{test}})$ exceeds $\maxrgtlv$. This emphasizes the importance of the calibration procedure presented in this section, which rigorously bridges the gap between the estimated maximal regret ${r}^{\textup{max}}_{\hat{\theta}_r}(b_{\text{test}})$ and the actual, unknown maximal regret $\rgt^{\textup{max}}_{\hat{\theta}}(v_{\text{test}};b_{\text{test}})$. 

In essence, we show how to form a corrected threshold on the output ${r}^\textup{max}_{\hat{\theta}_r}(b_{\text{test}})$, ensuring the regret of the auction outcome obtained from the auction model does not exceed the requested maximal regret level $\maxrgtlv$ with probability $1-\alpha$, resulting in a calibrated auction acceptance rule function. Given the auction model ${f}_{\hat{\theta}}$, a calibration set $\mathcal{D}_\textup{cal}$, a requested regret level $\maxrgtlv$, and a control level $\alpha$, we define the auction acceptance rule function $\hat{\mathcal{T}}_{\alpha}(b_\textup{test}, \maxrgtlv)$ for a test bid vector $b_{\text{test}}$ as follows: 

\begin{equation}
\label{eq:post_proccess_regret}
\hat{\mathcal{T}}_{\alpha}(b_{\text{test}},\maxrgtlv) =
    \begin{cases}
        {f}_{\hat{\theta}}(b_{\text{test}}), &  {r}^{\textup{max}}_{\hat{\theta}_r}(b_{\text{test}}) \leq \maxrgtlv - \hat{q}_{\alpha}, \\
        \;\;\;A_{0}, & \text{otherwise}.
    \end{cases}
\end{equation}

The function $\hat{\mathcal{T}}_{\alpha}(b_{\text{test}},\maxrgtlv)$ from above is determined by the estimated maximum regret ${r}^\textup{max}_{\hat{\theta}_r}(b_{\text{test}})$ and the calibration threshold $\hat{q}_{\alpha}$, which is derived using the calibration set $\mathcal{D}_\textup{cal}$. The output of the function $\hat{\mathcal{T}}_{\alpha}(b_{\text{test}}, \maxrgtlv)$ is the prediction of the auction model ${f}_{\hat \theta}({b_{\text{test}}})$ if ${r}^\textup{max}_{\hat{\theta}_r}(b_{\text{test}})$ is lower than the difference between the requested regret level $\maxrgtlv$, and the correction factor $\hat{q}_{\alpha}$ derived from our calibration method.  
Otherwise, the function $\hat{\mathcal{T}}_{\alpha}$ returns no-allocation rule $A_0$, indicating that the predicted auction does not meet the desired regret level. This acceptance rule ensures that the auction mechanism only operates within the predefined regret bounds.

Algorithm \ref{alg:regret_calibration} details the calibration process. The procedure starts by computing the regret scores $e_k = \rgt^{\textup{max}}_{\hat \theta}(v_k;b_k) - {r}^{\textup{max}}_{\hat{\theta}_r}(b_k)$, defined as the difference between the actual and the estimated regrets for each bidding $b_k$ and valuation $v_k$ pair of the calibration dataset. Then, we compute the empirical quantile $\hat{q}_{\alpha}$ of the regret scores $e_k$, $k=1,\ldots, K$. This quantile forms our threshold, which is used to set the corrected cutoff in the acceptance rule $\hat{\mathcal{T}}_{\alpha}$ from \eqref{eq:post_proccess_regret}. That is, $\hat{q}_{\alpha}$ rigorously accounts for the prediction error of the regret estimation model ${r}_{\hat{\theta}_r}$.

\begin{algorithm}[]
\caption{Regret Calibration}
\label{alg:regret_calibration}
\begin{algorithmic}[1]
    \State \textbf{Input:} ${f}_{\hat{\theta}}$ --  pre-trained auction model; $r_{\hat{\theta}_r}$ -- pre-trained regret model; $\maxrgtlv$ -- max regret level; $(1-\alpha)$ -- requested confidence level; $\mathcal{D}_\textup{cal}$ -- holdout calibration set.
    \State Initialize an empty list $R$ to store regret scores.
    \For{ $(b_k,v_k)$ in $\mathcal{D}_\textup{cal}$}
        \State Compute the estimated max regret ${r}_{\hat{\theta}_r}^\textup{max}(b_k)$.
        \State Compute the actual max regret $\rgt_{\hat{\theta}}^\textup{max}(v_k;b_k)$.
        \State Calculate $e_k = \rgt_{\hat{\theta}}^\textup{max}(v_k;b_k) - {r}_{\hat{\theta}_r}^\textup{max}(b_k)$.
        \State Append $e_k$ to $R$.
    \EndFor
    \State Compute  $\hat{q}_{\alpha}:=$ the $\lceil (1-\alpha)(|\mathcal{D}_\textup{cal}|+1)\rceil$-th smallest element in $\{e_i \in R\}\cup\{\infty\}$.
    \State \Return Calibrated acceptance rule function $\hat{\mathcal{T}}_{\alpha}(b, \maxrgtlv)$, defined in \eqref{eq:post_proccess_regret} and implemented with $\hat{q}_{\alpha}$.
\end{algorithmic}
\end{algorithm}

\begin{table*}[]
    \caption{Summary of the average revenue and empirical regret, with their standard deviations in parentheses, for both \texttt{RegretNet} and \ourmethod across $2\times2$, $2\times3$, and $3\times5$ auction settings. The test and the calibration sets are of size 1,000 examples. The `Requested Max Regret Level` column indicates the requested regret used in \ourmethod to control the regret at $99\%$ confidence ($\alpha = 0.01$). The `Accepted Auctions` column shows the percentage of accepted auctions by \ourmethod based on the regret threshold. the `Max Regret` column refers the max regret obtained at test time.\\}
    \centering
    \begin{tabular}{lcccccc}
        \toprule
         \makecell{Auction Model} & 
         \makecell{Auction Setup\\ Bidders × Items} &  
         \makecell{Revenue} & 
         \makecell{Empirical\\ Regret} & \makecell{Requested\\ Max Regret} & \makecell{Max\\ Regret} & \makecell{Acceptance\\ Auctions}\\
         \midrule
         \texttt{RegretNet} & $2\times2$       & 0.93 (0.33) & 0.007 (0.007) & 0.025 &  0.032 & 100\% \\
         \ourmethod [ours] & $2\times2$   & 0.90 (0.31) & 0.005 (0.004) & 0.025 & \textbf{0.024} &  93.9\%\\
         \cmidrule(lr){1-7}
         \texttt{RegretNet} & $2\times3$       & 1.38 (0.37) & 0.007 (0.008) & 0.035 &  0.045 & 100\% \\
         \ourmethod [ours] & $2\times3$   & 1.34 (0.38) & 0.006 (0.005) & 0.035 & \textbf{0.033} &  95.2\%\\
         \cmidrule(lr){1-7}
         \texttt{RegretNet} & $3\times5$       & 2.91 (0.40) & 0.012 (0.012) & 0.055 &  0.073 & 100\% \\
         \ourmethod [ours] & $3\times5$   & 2.60 (0.39) & 0.010 (0.009) & 0.055 & \textbf{0.054} &  97.7\%\\
         &  &  &  &  &  
    \end{tabular}
    \label{tab:results}
\end{table*}

We now prove that any test auction accepted by the rule $\hat{\mathcal{T}}_{\alpha}$ satisfies Definition~\ref{def:statistical_SP}. Simply put, this result guarentees that by setting $\alpha = 0.01$ and specifying a maximal regret level of $\maxrgtlv = 0.025$, the actual regret of new auctions accepted by $\hat{\mathcal{T}}_{\alpha}$ will not exceed 0.025 with probability 99\%.

\begin{prop}
\label{thm:statistical_SP}
Suppose the calibration set $\mathcal{D}_\textup{cal} = \{(b_k, v_k)\}_{k=1}^K$ and a new auction $(b_\textup{test},v_\textup{test})$ are i.i.d samples from $P_{b,v}$. Let $\hat{\mathcal{T}}_{\alpha}$ be the acceptance rule of a given auction model ${f}_{\hat{\theta}}$ and regret estimator $r_{\hat{\theta}_r}$, obtained from Algorithm \ref{alg:regret_calibration}. Then, the auction ${f}_{{\hat{\theta}}}(b_\textup{test})$ accepted by the rule $\hat{\mathcal{T}}_{\alpha}$ in \eqref{eq:post_proccess_regret} satisfies Definition~\ref{def:statistical_SP} with regret level $\maxrgtlv$ and confidence level $1-\alpha$.
\end{prop}

\section{Experiments}
\label{sec:experiments}
\label{sec:exp_setup}
In this section, we empirically evaluate our proposal on auctions with configurations of (bidders $\times$ items) sets of size $2\times2$, $2\times3$, and $3\times5$. The training process is based on the \texttt{RegretNet} implementation by \citet{dutting2019optimal} (PyTorch version by \cite{stein2023neural}), which is used as the data-driven auction model $f_{\hat{\theta}}$. Throughout the experiments, we refer to our conformal method as \ourmethod. That is, we apply the acceptance rule $\hat{\mathcal{T}}_{\alpha}$ to the output of the auction model  ${f}_{\hat{\theta}}(b) = [{g}_{\hat{\theta}}(b), {p}_{\hat{\theta}} (b)]$. We fit our regret estimation model ${r}_{\hat{\theta}_r}$ within the shared backbone architecture setup. 
For further implementation details, please refer to Appendix \ref{sec:Architectural_and_Training_Details}.

The performance of our \ourmethod method is evaluated in terms of revenue, empirical regret, and maximum regret. Table \ref{tab:results} summarizes the results for different auction configurations. Following that table, we can see our method achieves competitive revenue compared to \texttt{RegretNet} across different auction settings while controlling the requested levels of maximum regret. Observe how the maximum regret values obtained by \texttt{RegretNet} are uncontrolled.


\begin{figure}[t]
  \centering
  \begin{minipage}[t]{0.48\textwidth}
    \centering
    \includegraphics[height=4.52cm]{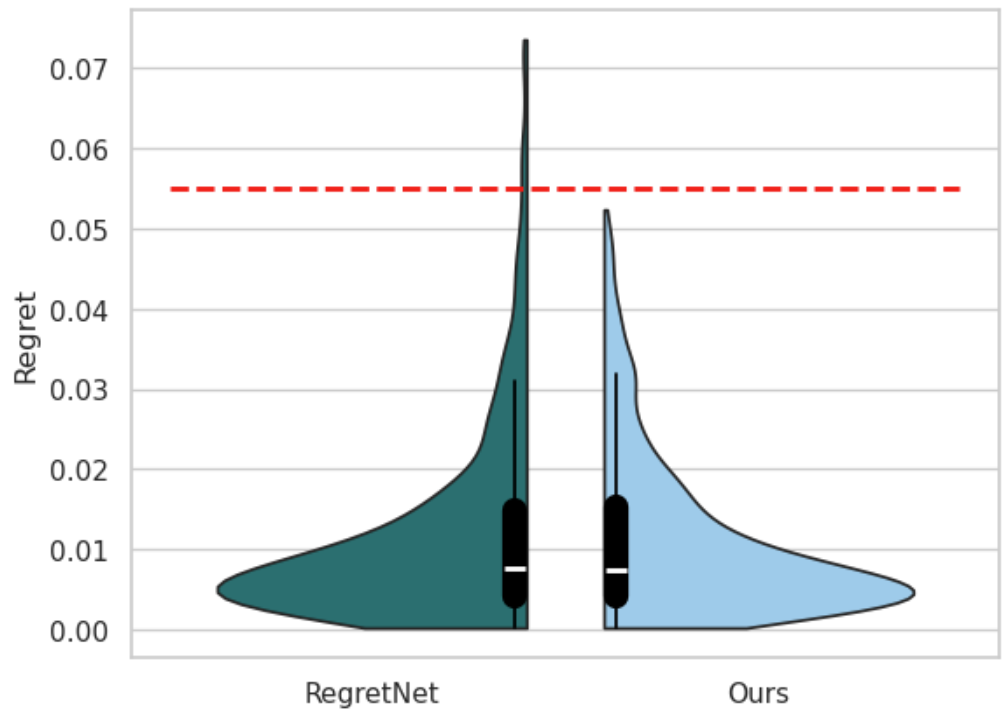}
    \captionof{figure}{Distribution of regret values for \texttt{RegretNet} and \ourmethod in the $2\times3$ auction setting. The test and calibration sets each contain 1,000 examples. The red dashed line indicates the requested regret level $\maxrgtlv$.}
    \label{fig:vp}
  \end{minipage}\hfill
  \begin{minipage}[t]{0.48\textwidth}
    \centering
    \includegraphics[height=4.52cm]{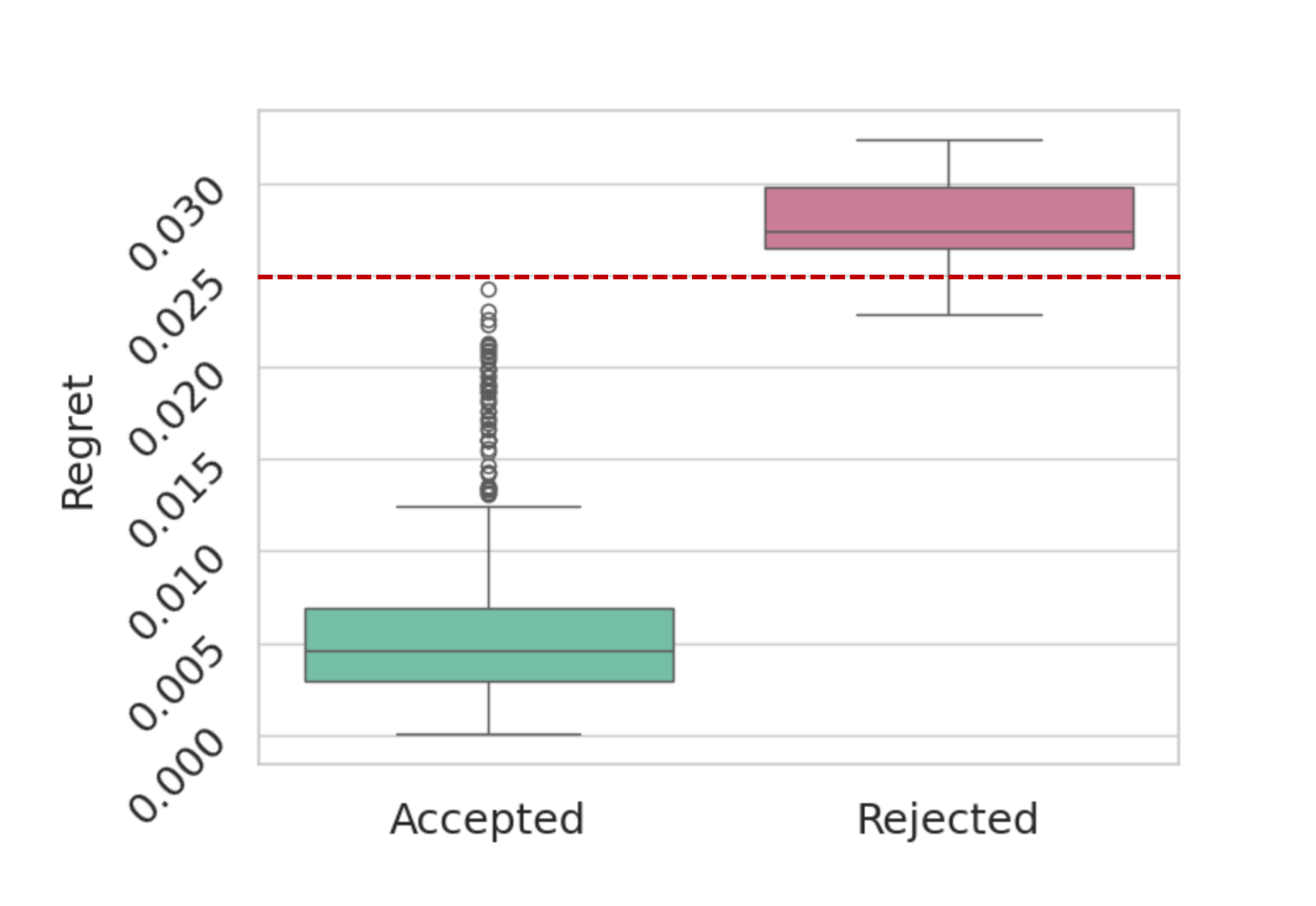}
    \captionof{figure}{Boxplot of the true regret values for \ourmethod in the $2\times3$ auction setting. The test and calibration sets each contain 1,000 examples. The red dashed line indicates the requested regret level $\maxrgtlv$.}
    \label{fig:bxp}
  \end{minipage}
\end{figure}



Table \ref{tab:results} also shows that our acceptance rule tends to accept most auctions at test time. Further, Figure \ref{fig:vp} illustrates that (i) the auctions rejected by our method are those that exceed the specified regret level; and (ii) the regret values of the accepted auctions fall below the maximal regret level $\maxrgtlv$. Consistently, Fig.~\ref{fig:bxp} shows rejected auctions are concentrated above $\maxrgtlv$, while accepted auctions are concentrated below it. Taken together, these results imply that \ourmethod controls regret, rejects high-regret auctions, and tends not to reject low-regret auctions---i.e., it accepts the auctions that should be accepted.

\begin{table*}[]
    \caption{Comparison of average revenue and strategy-proofness of various auction models in the $2\times2$ auction setting. The table includes our \ourmethod compared to classical mechanisms (\texttt{VCG}, \texttt{Item-Myerson}), transformer-based models (\texttt{AMenuNet}, \texttt{RegretFormer}), and other deep learning-based approaches (\texttt{RegretNet}, \texttt{EquivariantNet}, \texttt{ALGNet}). Strategy-proofness is indicated as either Yes, No, or Statistically. The results presented in this table are as reported by the authors for the same setup.}
    
    \centering
    \begin{tabular}{lcccc}
         \toprule
         \makecell{Auction Model} & 
         Revenue & 
         \makecell{Strategy-Proof}\\
         \midrule
         \texttt{VCG} \citep{vickrey1961counterspeculation}              & 0.28 & Yes \\
         \texttt{Item-Myerson} \citep{Myerson81}     & 0.42 & Yes \\
         \texttt{AMenuNet} \citep{duan2024scalable}    & 0.43 & Yes \\
         \texttt{RegretNet} \citep{dutting2019optimal}      & 0.93  & No \\
         \texttt{EquivariantNet} \citep{rahme2021permutation}  & 0.87  & No \\
         \texttt{ALGNet} \citep{rahme2021auction}         & 0.88  & No \\
         \texttt{RegretFormer} \citep{ivanov2022optimal}  & 0.94  & No \\
         \ourmethod \:(ours)  & 0.92 & Statistically \\
         
         &  &  &   &
    \end{tabular}
    \label{tab:results_rev}
\end{table*}

Lastly, Figure \ref{fig:rejected_items_vs_requested_regret_level} presents the effect of the maximal regret level $\maxrgtlv$ on the number of auctions rejected by $\hat{\mathcal{T}}_{\alpha}$. Following that figure, we can see that as the auctioneer enforces a lower max regret level $\maxrgtlv$, the number of rejected auctions increases as they fail to meet this  stricter criterion.

\begin{figure}[]
    \centering
    \includegraphics[height=4.52cm]{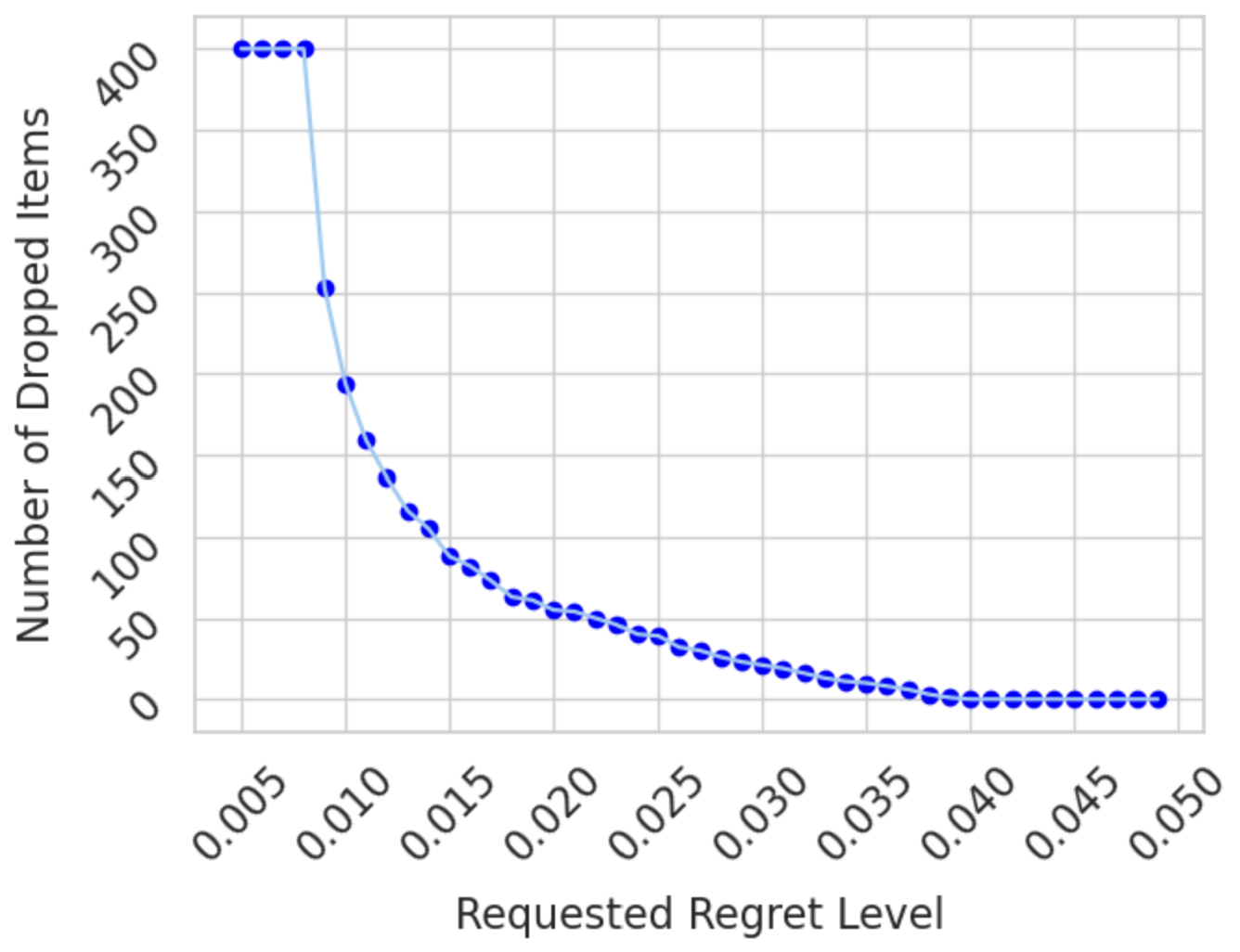}
    \caption{Effect of $\maxrgtlv$ on the number of auctions rejected by the calibrated acceptance rule $\hat{\mathcal{T}}_{\alpha}(b,\maxrgtlv)$ in the $2\times2$ auction setup. Evaluated on a test set of 400 examples.}
    \label{fig:rejected_items_vs_requested_regret_level}
\end{figure}

Overall, the results validate the effectiveness of our approach in designing auctions that balance revenue maximization with approximate strategy-proofness.

\section{Discussion}
\label{sec:discussion}

In this paper, we introduced the notion of statistical strategy-proof auctions, and  formulated a data-driven acceptance rule that selectively accepts auctions, at test time, whose regret is guaranteed to fall below a user-specified maximum regret value with high probability. The formulation of our data-driven acceptance rule is achieved based on predictions obtained by our novel regret estimation model, which is then calibrated by building on ideas from the conformal prediction literature. 
Numerical experiments showed that while the average regret of the vanilla \texttt{RegretNet} model is low, its maximum regret can be relatively high. By applying our proposed acceptance rule, the regret of accepted auctions consistently remained within the prescribed maximum regret level. Notably, the results demonstrated the effectiveness of our approach in filtering out high-regret auctions while retaining those with low regret. The experiments also demonstrated that our approach achieved notably higher revenue compared to baseline methods that enforce strictly strategy-proof auctions with zero regret, highlighting the practical relevance of our statistical approach.

The usefulness of our acceptance rule is affected by the predictive power of the underlying regret estimation model. If this model has poor predictive capabilities, our algorithm may make more rejections to ensure that the maximal regret will be controlled. Future research could benefit from exploring various learning schemes to predict the unknown regret more effectively.

A natural limitation of our method is the i.i.d.~assumption we make, which might be violated in practice. This may happen due to changes in the bidding and valuation distribution profiles, or due to imperfect calibration data that include noisy or inaccurate valuations, for instance.  As a future research direction, it would be important to quantify the effect of violation of this assumption on the resulting maximal regret level, possibly by taking inspiration from \cite{barber2023conformal}. It would be also intriguing to design calibration frameworks that are robust to such distribution shifts in the data.

Another worthwhile direction for future work is to extend our framework to simultaneously control the test-time revenue and regret. We believe this could be done by borrowing ideas from \citep{angelopoulos2021learn}, which presents a general calibration framework for controlling multiple risks. Such an approach would enable the auctioneer to ensure, with high probability, that the calibrated auction mechanism maintains low regret while also keeping the revenue of the acceptance auctions above a specified threshold. This multi-objective control is particularly relevant for large-scale, multi-item auctions, where revenue could potentially fall dramatically due to a large number of items.

\section*{Acknowledgements}                         
Y.R. and R.M.L. were supported by the ISRAEL SCIENCE FOUNDATION (grant No. 729/21). Y.R. thanks the Career Advancement Fellowship, Technion. I.T.C received funding from the European Research Council (ERC) under the European Union’s Horizon 2020 research and innovation program (grant agreement: 101077862, project name ALGOCONTRACT, PI: Inbal Talgam-Cohen).

\bibliography{ref}

\appendix
\appendix
\clearpage
\onecolumn 
\section{Additional Experiments}
This appendix provides additional experimental results and diagnostics that complement the main text. We report extended quantitative metrics—coverage of the statistical truthfulness constraint, realized regret relative to the requested level $\maxrgtlv$, and revenue—alongside robustness checks to calibration set size, model capacity, and distribution shift, and ablations of the components of \ourmethod. Unless stated otherwise, the data generation process, training protocol, and hyperparameters are identical to those in Sec.~\ref{sec:experiments}.

Table~\ref{tab:results_big} mirrors Table~\ref{tab:results} but evaluates on larger calibration and test sets (10,000 each). The pattern is unchanged: \ourmethod achieves revenue close to \texttt{RegretNet} while markedly lowering empirical regret and keeping test-time max regret at or below the requested level $\maxrgtlv$ in all settings. By contrast, \texttt{RegretNet} exceeds the requested max-regret level in each case. The increased sample size tightens the estimates and demonstrates that \ourmethod scales to large calibration and test sets while preserving coverage and revenue.

\begin{table*}[ht]
    \caption{Summary of the average revenue and empirical regret, with their standard deviations in parentheses, for both \texttt{RegretNet} and \ourmethod across $2\times2$, $2\times3$, and $3\times5$ auction settings. The test and the calibration sets are of size 10,000 examples. The `Requested Max Regret Level` column indicates the requested regret used in \ourmethod to control the regret at $99\%$ confidence ($\alpha = 0.01$). The `Accepted Auctions` column shows the percentage of accepted auctions by \ourmethod based on the regret threshold. the `Max Regret` column refers the max regret obtained at test time.\\}
    \centering
    \begin{tabular}{lcccccc}
        \toprule
         \makecell{Auction Model} & 
         \makecell{Auction Setup\\ Bidders × Items} &  
         \makecell{Revenue} & 
         \makecell{Empirical\\ Regret} & \makecell{Requested\\ Max Regret} & \makecell{Max\\ Regret} & \makecell{Acceptance\\ Auctions}\\
         \midrule
         \texttt{RegretNet} & $2\times2$       & 0.93 (0.33) & 0.007 (0.006) & 0.025 &  0.039 & 100\% \\
         \ourmethod [ours] & $2\times2$   & 0.92 (0.32) & 0.004 (0.003) & 0.025 & 0.020 &  89.23\%\\
         \cmidrule(lr){1-7}
         \texttt{RegretNet} & $2\times3$       & 1.41 (0.38) & 0.011 (0.014) & 0.035 &  0.164 & 100\% \\
         \ourmethod [ours] & $2\times3$   & 1.38 (0.39) & 0.007 (0.005) & 0.035 & 0.032 &  94.4\%\\
         \cmidrule(lr){1-7}
         \texttt{RegretNet} & $3\times5$       & 2.93 (0.40) & 0.023 (0.012) & 0.055 &  0.121 & 100\% \\
         \ourmethod [ours] & $3\times5$   & 2.67 (0.38) & 0.009 (0.007) & 0.055 & 0.055 &  93.4\%\\
         &  &  &  &  &  
    \end{tabular}
    \label{tab:results_big}
\end{table*}

To demonstrate the flexibility of our approach, we now conduct experiments in which we treat the original \texttt{RegretNet}  model as a black-box and form an acceptance rule by training a separate regret model. The results, summarized in Table~\ref{tab:results_blackbox}, show the validity of our approach and our ability to perform well even without having access to the auction model's parameters.

\begin{table}[ht]
    \caption{Summary of the average revenue and empirical regret (with standard deviations in parentheses) for \texttt{RegretNet}, \ourmethod-shared (\ref{fig:regretnet_with_score_model}), and \ourmethod-separate (\ref{fig:regretnet_with_sep_score_model}) in the black-box settings described in Section~\ref{sec:regret_estimation}. Each model’s requested and actual maximum regret is reported, along with the percentage of auctions accepted under the requested regret threshold. Test and calibration sets each contain 10,000 examples.}
    
    \centering
    \begin{tabular}{lcccccc}
        \toprule
         \makecell{Auction Model} & 
         \makecell{Auction Setup\\ Bidders × Items} & 
         \makecell{Revenue} & 
         \makecell{Empirical\\ Regret} & \makecell{Requested\\ Max Regret} & \makecell{Max\\ Regret} & \makecell{Acceptance\\ Auctions}\\
         \midrule
         \texttt{RegretNet} & $2\times2$       & 0.93 (0.33) & 0.007 (0.006) & 0.025 &  0.039 & 100\% \\
         \ourmethod-shared (\ref{fig:regretnet_with_score_model}) &  $2\times2$   & 0.92 (0.32) & 0.004 (0.003) & 0.025 & 0.020 &  89.23\%\\
        \ourmethod-separate (\ref{fig:regretnet_with_sep_score_model}) &  $2\times2$  & 0.92 (0.31) & 0.004 (0.005) & 0.025 & 0.023 &  93.7\%\\
         &  &  &  &  & 
    \end{tabular}
    \label{tab:results_blackbox}
\end{table}


Table~\ref{tab:results_rev_full} mirrors Table~\ref{tab:results_rev} and expands it with additional bidder–item configurations, showing that \ourmethod scales cleanly with both the number of bidders and items. Across all settings, \ourmethod attains revenue that is competitive with the strongest learning-based baselines. Unlike these high-revenue DNN baselines, which are not strategy-proof, \ourmethod preserves statistical strategy-proofness at test time. Classical truthful mechanisms provide strict strategy-proofness but generally trail the top non-truthful models in revenue for larger setups. Overall, the extended results confirm a favorable revenue–truthfulness trade-off: \ourmethod maintains near–state-of-the-art revenue while delivering test-time guarantees and scaling gracefully as the auction size grows.

\begin{table}[H]
    \caption{Comparison of average revenue and strategy-proofness of various auction models in different auction settings. The table includes our \ourmethod compared to classical mechanisms (\texttt{VCG}, \texttt{Item-Myerson}), transformer-based models (\texttt{AMenuNet}, \texttt{RegretFormer}), and other deep learning-based approaches (\texttt{RegretNet}, \texttt{EquivariantNet}, \texttt{ALGNet}). Strategy-proofness is indicated as either Yes, No, or Statistically. The results presented in this table are as reported by the authors for the same setups.}
    
    \centering
    \begin{tabular}{lccccc}
    \toprule
    \makecell{Auction Model} &
    \makecell{1x2} &
    \makecell{2x2} &
    \makecell{2x3} &
    \makecell{3x10} &
    \makecell{Strategy-Proof} \\
    \midrule
    \texttt{VCG}~\citep{vickrey1961counterspeculation} & 0.00 & 0.28 & 1.00 & 5.02 & Yes \\
    \texttt{Item-Myerson}~\citep{Myerson81} & 0.14 & 0.42 & 1.25 & 5.31 & Yes \\
    \texttt{AMenuNet}~\citep{duan2024scalable} & 0.17 & 0.43 & 1.30 & 5.58 & Yes \\
    \texttt{RegretNet}~\citep{dutting2019optimal} & 0.57 & 0.93 & 1.40 & 5.59 & No \\
    \texttt{EquivariantNet}~\citep{rahme2021permutation} & 0.58 & 0.87 & 1.42 & 5.74 & No \\
    \texttt{ALGNet}~\citep{rahme2021auction} & 0.55 & 0.88 & 1.44 & 5.52 & No \\
    \texttt{RegretFormer}~\citep{ivanov2022optimal} & 0.57 & 0.94 & 1.48 & 6.12 & No \\
    \ourmethod~(ours) & 0.57 & 0.92 & 1.39 & 5.52 & Statistically \\
    \bottomrule
    
    \end{tabular}
\label{tab:results_rev_full}
\end{table}

Figure \ref{fig:joint_dist_regret_regrethat} presents the actual regret values versus the predicted regret values from our regret estimation model ${r}_{\hat{\theta}_r}$. The strong correlation observed between them suggests that the model predicts regret with reasonable accuracy. This aligns with the behavior seen in Figure~\ref{fig:vp}, where our method consistently rejects auctions with high regret.

\begin{figure}[H]
    \centering
    \includegraphics[width=0.47\textwidth]{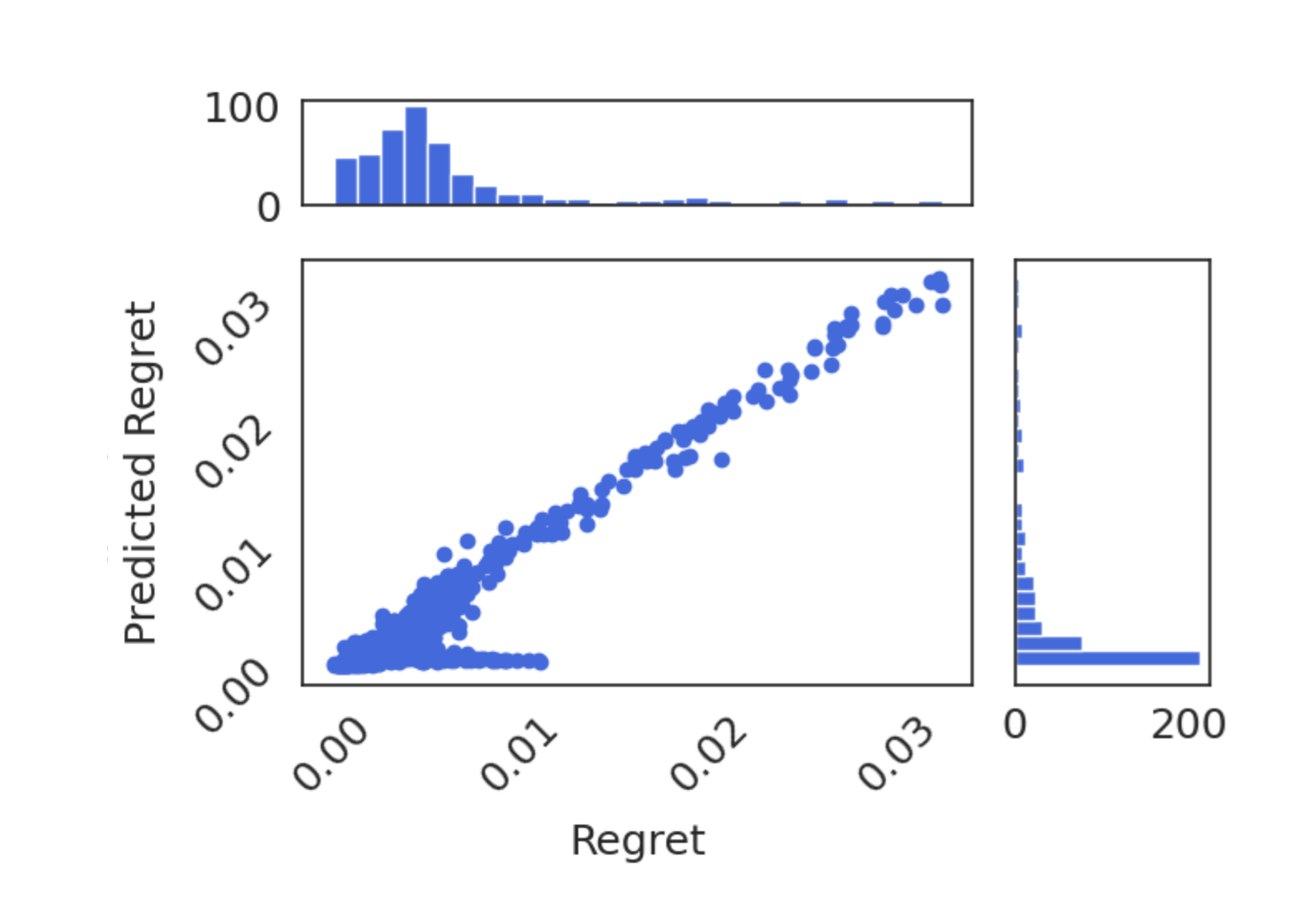}
    \caption{Scatter plot of the actual regret $\rgt_{\hat{\theta}}(v;b)$ and the predicted regret ${r}_{\hat \theta_r}(b)$ of \ourmethod, evaluated on test data of size $1,000$; the auction setup is $2\times2$.}
    \label{fig:joint_dist_regret_regrethat}
\end{figure}

\section{Mathematical Proofs}
\begin{proof}
\emph{of Theorem~\ref{thm:statistical_SP}}. To lighten notations, denote by $r$ the actual maximum regret $\rgt^\textup{max}_{\hat{\theta}}(v_\textup{test};b_\textup{test})$ and by $\hat{r}$ the predicted maximum regret ${r}_{\hat{\theta}_r}^\textup{max}(b_\textup{test})$. With these notations in place, the probability of interest is as follows:
\begin{equation*}
    \begin{aligned}
&\mathbb{P}\left[(r > \maxrgtlv) \:\cap\:  (\hat{\mathcal{T}}_{\alpha}(b_{\textup{test}},\maxrgtlv)\neq A_0)\right]\\ 
    &=\:\mathbb{P}\left[(r > \maxrgtlv) \:\cap\: (\hat{r}<\maxrgtlv-\hat{q}_{\alpha})\right] &\textit{\textcolor{gray}{by definition of \,$\hat{\mathcal{T}}_{\alpha}$}}\\
    &=\:\mathbb{P}\left[(r > \maxrgtlv) \:\cap\: (\maxrgtlv>\hat{r}+\hat{q}_{\alpha})\right] \\
    &=\:\mathbb{P}\left[r > \hat{r}+\hat{q}_{\alpha}\right] 
    &\textit{\textcolor{gray}{intersecting events}}\\
    &=\:\mathbb{P}\left[r -\hat{r} > \hat{q}_{\alpha}\right] \\
    &=\:1-\mathbb{P}\left[r -\hat{r} \leq \hat{q}_{\alpha}\right] 
    &\textit{\textcolor{gray}{complement event}}\\
    &\geq\:1-(1-\alpha) \\
    &=\:\alpha,
    \end{aligned}
\end{equation*}
where the last inequality holds due  \citep[][Lemma 2]{romano2019conformalized}.
\end{proof}

\section{Experiment setup and environment}
The experiments were conducted on a high-performance computing setup consisting of an Ubuntu 20.04.6 LTS operating system, powered by 96 Intel(R) Xeon(R) Gold CPUs running at 2.40 GHz. The machine was equipped with 16 Nvidia A40 GPUs and 512 GB of RAM. The software environment was based on PyTorch 2.1, running under Python 3.11.5.

\section{Architectural and Training Details}
\label{sec:Architectural_and_Training_Details}

The architecture of our models includes ReLU activations for both the allocation and payment networks, as well as the regret estimation model if exists. To ensure that each item's allocations sum to one, we apply a softmax activation to the allocation network. For the payment network, a sigmoid activation function is used to constrain the payments to a fraction between zero and one, enforcing individual rationality (IR).

Models are trained on 700,000 sample valuation and bids profiles $v_{i,j}\sim U[0,1]$, $b_{i,j}\sim U[0,v_{i,j}]$, divided into batches, using the Adam optimizer with an initial learning rate of $1e^{-3}$. The training process spans several epochs, with the exact number described in Table \ref{tab:app-hyperparams}.

An additional term for regret minimization is incorporated into the training process, similar to the augmented Lagrangian method used in \texttt{RegretNet}. The weight for this regret term is periodically updated to adjust the emphasis on regret minimization in the loss function. Furthermore, the value of the parameter $\rho$, which influences the regret term, is also periodically updated during training.

The agents' regret is optimized in an inner loop, with specific parameters detailed in the last three rows of Table \ref{tab:app-hyperparams}. For a comprehensive understanding of the loss function derivation and further details on the training procedure, we refer the reader to the original RegretNet paper by \citet{dutting2019optimal}.

\begin{table}[H]
\centering
\caption{\textbf{Hyperparameters.} We provide training hyperparameters for all auction sizes we study.\\}
\label{tab:app-hyperparams}
\begin{tabular}{lccc}
\toprule
Auction setup (Bidders x Items) & 2$\times$2 & 2$\times$3 & 3$\times$5 \\
\midrule
Epochs                                & 50   & 50   & 50   \\
Train Batch Size                      & 512  & 1024  & 2048  \\
Initial Learning Rate                 & 0.001 & 0.005 & 0.01 \\
Number of Hidden Layers               & 5    & 5    & 5    \\
Hidden Layer size                     & 100  & 100  & 100  \\
$\rho$ Update Period (Epochs)         & 2    & 2    & 2    \\
Lagrange Weight Update Period (Iters) & 100  & 100  & 100  \\
Initial $\rho$                        & 1    & 1    & 1    \\
$\rho$ Increment                      & 1   & 5    & 8    \\
Initial Lagrange Weight               & 5    & 5    & 5    \\
Misreport Learning Rate (Training)    & 0.1  & 0.1  & 0.1  \\
Misreport Iterations (Training)       & 25   & 25   & 25   \\
Misreport Initializations (Training)  & 10   & 10   & 10   \\
\bottomrule
\end{tabular}
\end{table}

\section{Regret Prediction Model Illustration}
\label{asec:regret_model_illustration}
\begin{figure}[H]
    \centering
    \begin{subfigure}[b]{0.45\textwidth}
        \centering
        \includegraphics[width=\textwidth]{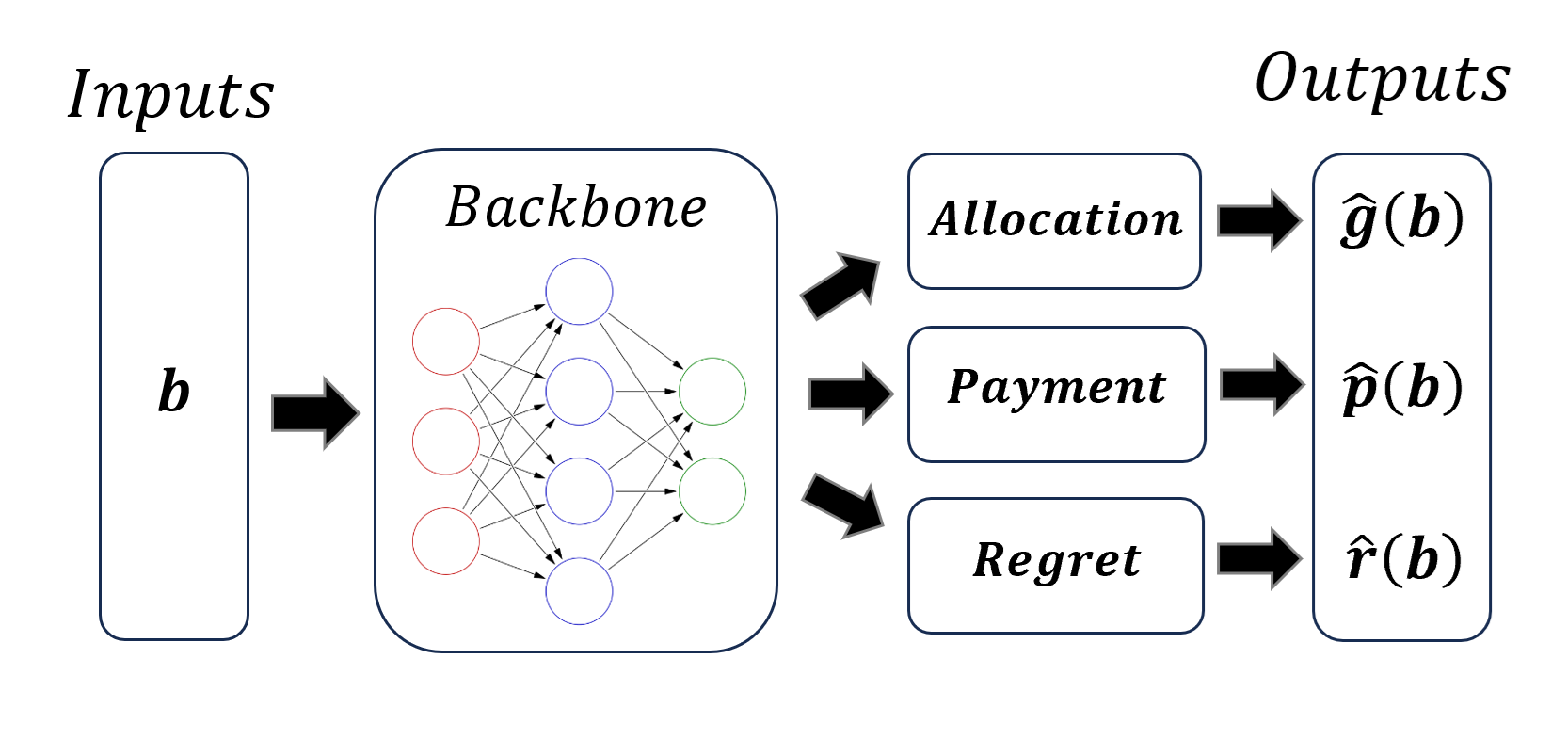}
        \caption{Shared backbone regret estimation.}
        \label{fig:regretnet_with_score_model}
    \end{subfigure}
    \hfill
    \begin{subfigure}[b]{0.45\textwidth}
        \centering
        \includegraphics[width=\textwidth]{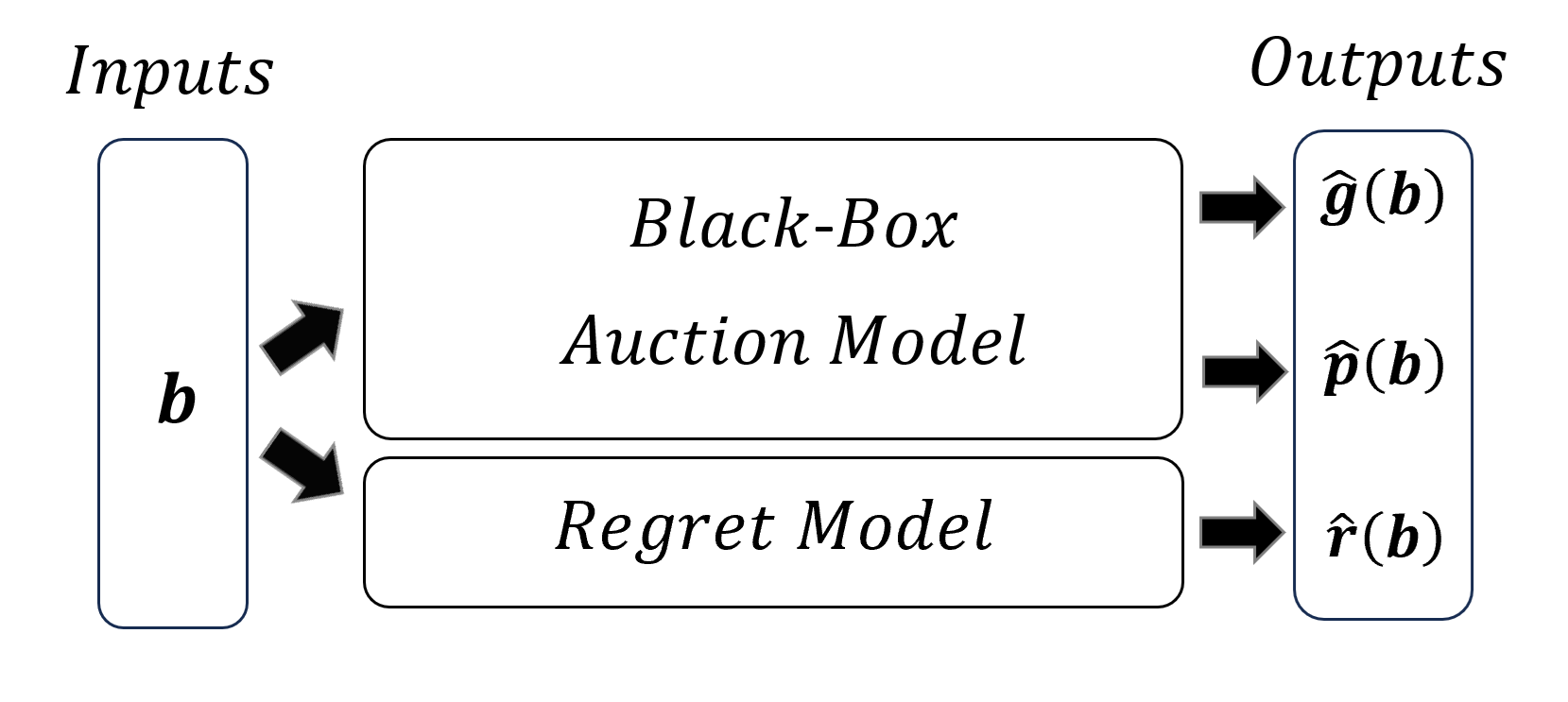}
        \caption{Black-box auction model regret estimation.}
        \label{fig:regretnet_with_sep_score_model}
    \end{subfigure}
    \caption{Comparison of regret estimation architectures in auction models.}
    \label{fig:regretnet_models_combined}
\end{figure}

\end{document}